\documentclass[12pt,a4paper,titlepage]{article}

\usepackage{indentfirst} %% Отсуп у первой строки раздела
\usepackage{lastpage} %% Количество страниц

\usepackage{amsmath} 
\usepackage{amssymb}
\usepackage{amsfonts}
\usepackage{amsthm}
\usepackage{bm} % Делает любые символы жирнными
\usepackage{icomma} % "Умная" запятая: $0,2$ --- число, $0, 2$ --- перечисление

\usepackage{extsizes} % Правильно меняет размер тескта
\usepackage{geometry} % Поля
\usepackage{setspace}
\usepackage{textcase}

\usepackage{cancel}

\usepackage[linkcolor=blue,colorlinks=true]{hyperref}

\newtheorem{theorem}{Theorem}

\theoremstyle{definition}
\newtheorem{definition}{Definition}

\usepackage{epsfig}
\usepackage{graphicx}

\begin{document}

	\begin{center}
		\Large
	\textbf{Repeated temperature measurements in quantum thermodynamics}
	
		\large 
    \textbf{N.M. Gerasimov}\footnote{Faculty of Physics, Lomonosov Moscow State University,
				Leninskie Gory, Moscow 119991, Russia\\ E-mail:\href{mailto:gerasimov.nm20@physics.msu.ru}{gerasimov.nm20@physics.msu.ru}},
		\textbf{A.E. Teretenkov}\footnote{Department of Mathematical Methods for Quantum Technologies, Steklov Mathematical Institute of Russian Academy of Sciences,
			ul. Gubkina 8, Moscow 119991, Russia\\ E-mail:\href{mailto:taemsu@mail.ru}{taemsu@mail.ru}}
		\end{center}

	\begin{center}
	In this work, we model the temperature measurement as a transformation of the arbitrary state into the Gibbs state. We start with a general formalism of ansatz-posteriors,  which includes many usual models of posterior states due to measurement or state preparation. On the one hand, it contains models of selective and non-selective measurements posteriors as special cases, which allows us to consider it as a generalization of usual measurement models. On the other hand, it contains the above mentioned transformation into the Gibbs state as a special case as well. We derive an analogue of the master equation in the stroboscopic limit of repeated measurements. Then we apply our general approach to  temperature measurement.
	\end{center}

    \section{Introduction}

Temperature measurement is one of the main aims of quantum thermodynamics. The description of temperature, the mechanism of its measurement and invasiveness of thermometry \cite{Albarelli2023Invasiveness} play a key role in the development of quantum technologies.    The idea of probe thermometry is based on the measurement of the parameters of the sample via its interaction with a probe. Thus, the probe plays the role of an open quantum system connected to the sample which is considered the environment. Our goal is to describe the mechanism of repeated temperature measurements. The advantages of sequential measurements were discussed in \cite{DePasquale2017Sequential}.

We  consider temperature measurement similarly to selective von Neumann measurements as transformation of an arbitrary prior state to a posterior state with fixed temperature. And a natural choice for the state with fixed temperature is the Gibbs state. It will differ our generalization from usual von Neumann measurements, for which the states with a fixed result are described by an orthogonal projector.  Remark that we will focus on state dynamics under repeated temperature measurements, but we will not discuss the result probabilities.  Let us just mention that due to the fact that the elements of the convex hull of the Gibbs states admit non-unique decomposition into its extreme this hull, they fit the general probability theory rather than the usual one \cite{Sparaciari2017}.

In Section~\ref{sec:ansatzes} we introduce the general concept of ansatz from \cite{Meretukov2024} starting from its abstract definition. Although this abstract definition is too general for our specific goal, it is very important for making clear many useful analogies. In particular, it allows us to emphasize the similarities between our approach and usual stroboscopic limit of repeated measurements \cite{Luchnikov2017}. It also reveals the analogy of the stroboscopic limit approach in this work and the weak coupling limit approach in  \cite{Meretukov2024} based on Kawasaki-Gunton projectors. And the Kawasaki-Gunton projectors play an important role for derivation of dynamical equations in quantum thermodynamics \cite{Rau1996, Semin2015, Semin2020}.  

In Section~\ref{sec:posteriors} we introduce ansatz-posteriors and discuss how such a concept covers usual measurement and state preparation models. In Section~\ref{sec:strob} we derive an analog of the master equation in the stroboscopic limit of repeated measurements, this is one of the main results of this work. In Section~\ref{sec:tempMeas} we discuss the special case of generailized Gibbs ansatzes. In Section~\ref{eq:example} we consider examples of master equations in the stroboscopic limit of repeated measurements correspondent to  the Gibbs ansatz. Namely, we consider the repeated measurements of qubit and generic multilevel open systems in fixed temperature environments, whose dynamics between measurements is described by Gorini--Kossakowski--Sudarshan--Lindblad (GKSL) equations \cite{Gorini1976, Lindblad1976}. So they can be considered as models of indirect temperature measurements.

In Conclusions we summarize the results of the work and discuss directions of the further development.

	\section{Ansatzes and their parametrization}\label{sec:ansatzes}

In this work we consider the finite dimensional case only. So we recall some definitions valid for arbitrary $C^*$-algebras restricting them to $\mathbb{C}^{d \times d}$ (with usual Hermitian conjugation $^\dagger$ as *-involution).

\begin{definition}
	We call $G$  an operator system in $\mathbb{C}^{d \times d}$ if $G$ is a linear subspace of  $\mathbb{C}^{d \times d}$  such that:
	\begin{enumerate}
		\item $I \in G$.
		\item If $g \in G$, then $g^{\dagger} \in G.$
	\end{enumerate}
\end{definition}

Sometimes operator systems \cite{Choi1977, Paulsen2003, Yashin2022} are also called non-commutative graphs \cite{Duan2012, Amosov2019, Amosov2020, Amosov2021} or quantum graphs \cite{Weaver2021, Brannan2022, Priyanga2023}. But this term was inspired by confusability graph of classical channels, and our usage will be different and not naturally associated with any quantum channel, so we will use more application-neutral term ''operator system''.

\begin{definition}
	We call a continuous linear map $\varepsilon: G \rightarrow \mathbb{C}$  a state on operator system $G$ in $\mathbb{C}^{d \times d}$ if
	\begin{enumerate}
		\item $\varepsilon(g^{\dagger}) = (\varepsilon(g))^*$.
		\item $\varepsilon(I) = 1$.
		\item If $g \geqslant 0$, then $\varepsilon(g) \geqslant 0$.
	\end{enumerate}
\end{definition}

Now we can define an ansatz.
\begin{definition}
	Let $G$ be an operator system in $\mathbb{C}^{d \times d}$. Let $S_{\varepsilon} $ be a family of states on $\mathbb{C}^{d \times d}$   parameterized by states $\varepsilon$ on $G$ such that
	\begin{equation}\label{eq:abstrConsistCond}
		S_{\varepsilon}(g)  = \varepsilon(g)
	\end{equation}
	for all $g \in G$. Then we call $G$  a relevant operator system, we call $S_{\varepsilon}$ an ansatz consistent with relevant observables, and we call  \eqref{eq:abstrConsistCond} the consistency conditions.
\end{definition}

In practice, for the finite-dimensional case  it is natural to use the representation
\begin{equation}\label{eq:normalStatesParam}
	S_{\varepsilon}(X) =  \operatorname{Tr} (X \rho_{ans}( \varepsilon))
\end{equation}
for all $X \in \mathbb{C}^{d \times d}$.

For many open quantum system applications  it is natural to consider an operator system of the form $\mathbb{C}^{d_S \times d_S} \otimes I_B$ in $\mathbb{C}^{d_Bd_S \times d_Bd_S} \simeq \mathbb{C}^{d_S \times d_S} \otimes \mathbb{C}^{d_B \times d_B}$, which is not just  an operator system, but a matrix *-algebra \cite{Klerk2011} *-isomorphic to the full $\mathbb{C}^{d_S \times d_S}$ *-algebra. In this case, it is useful to parameterize the state on $\mathbb{C}^{d_S \times d_S} \otimes I_B \simeq \mathbb{C}^{d_S \times d_S}$ similarly to \eqref{eq:normalStatesParam} as
\begin{equation*}
	\varepsilon(X_S \otimes I_B) =   \operatorname{Tr} (X_S \rho_S),
\end{equation*}
where $\rho_S \in  \mathbb{C}^{d_S \times d_S}$ is a density matrix. Then one can reparameterize $ \rho_{ans}( \rho_S)$ in terms of new parameter $\rho_S$ instead of  $ \rho_{ans}(  \varepsilon)$. The most important ansatz for open quantum systems is
\begin{equation}\label{eq:factorizedAnsatz}
	\rho_{ans}( \rho_S) = \rho_S \otimes \rho_B,
\end{equation} 
where $ \rho_B$ is a fixed density matrix of the reservoir. But for many-body problems the non-linear ansatzes of the form
\begin{equation*}
	\rho_{ans}( \rho_S) = \rho_S \otimes \ldots \otimes \rho_S
\end{equation*} 
occur \cite[Section 3.7]{Breuer2002}. In general for the operator system of the form $\mathbb{C}^{d_S \times d_S} \otimes I_B$ such ansatzes are usually called assignment \cite{Pechukas1994, Alicki1995, Jordan2004, Rodriguez2010}
or recovery maps \cite{Trushechkin2021}.

Another example is a multi-mode fermionic even Gaussian state. Its matrix of second moments \cite{Teretenkov2019} defines a continuous linear positive functional on the linear span of all posible pairwise products of fermionic creation and annihilation operators. This span becomes an operator system by adding operators proportional to identity, and the functional becomes a state on it assuming $\varepsilon(I) = 1$. On the physical level of rigor one can think similarly about covariance matrices and vectors of means for multi-mode bosonic Gaussian states \cite{Meretukov2025}.

Sometimes the parametrization for linear ansatzes is not-explicit and based on projections $\mathcal{P}: \mathbb{C}^{d_S \times d_S} \rightarrow \mathbb{C}^{d_S \times d_S}$ such that $\mathcal{P} \rho$ is a density matrix for any density matrix $\rho$. Then
\begin{equation}\label{eq:projAns}
	\rho_{ans}(\varepsilon) = \mathcal{P} \rho,
\end{equation}
where  $\varepsilon$ is parameterized by the image of $\mathcal{P} \rho$.

For the most of thermodynamical applications it is useful to choose a self-adjoint basis $P_m$, $m = 0,1, \ldots, M$ in $G$ such that $P_0=I$, then the  states on   $G$ can be parameterized by 
\begin{equation*}
	E_m = \varepsilon(P_m), \qquad m =1, \ldots, M,
\end{equation*}
which will arrange in the vector form  
\begin{equation}\label{eq:vectorParam}
	\vec{E} = \varepsilon( \vec{P}).
\end{equation}
We  will also write $ \rho_{ans}( 	\vec{E})$ instead of $\rho_{ans} (\varepsilon)$. The elements of $\vec{P}$ we call the relevant observables. Then our notation will coincides with \cite{Meretukov2024}.  It is useful for thermodynamical applications, so we use this notation in the further sections. Nevertheless, it breaks many useful analogies so we started this section with the abstract coordinate-free definition.  Let us  remark that in physical literature  notation $\langle P_m \rangle$ for $E_m$ is widely used \cite{Zubarev1997, Semin2015, Semin2020}, but it hides the facts that $\langle P_m \rangle$ are parameters of the ansatz rather than just averages of relevant observables.

\section{Ansatz-posteriors as measurement and state preparation models}
\label{sec:posteriors}

There are two archetypical scenarios of ansatze-usage for master equation derivation. The first one is perturbative solution of equations which preserve the  ansatz in the zero-order of perturbation theory. In a very general setup this scenario can be played via generalized Kawasaki-Gunton projectors and was discussed in \cite{Meretukov2024}. The second scenario is that the transformation of a given state to ansatz state can be realized by some experimental intervention. This is the case for measurement procedures, where it models the posterior states arising after measurements, and effectively in collision models \cite{Campbell2021, Ciccarello2022}, where we assume that we can prepare new factorized states with the given system state by physically ''colliding'' our system with new but similar reservoirs. In this work we will focus on this second scenario.

We  model the measurement (or preparation) influence on the system of interest by transformation of the arbitrary prior state to a state from the ansatz family. It can be formalized by the following definition.

\begin{definition}\label{def:posteriors}
	Let $s$ be a state on $\mathbb{C}^{d \times d}$ and $ S_\varepsilon$ be an ansatz  on $\mathbb{C}^{d \times d}$ with relevant operator system $g$, then the ansatz-posterior is $ S_\varepsilon$ for $\varepsilon$ chosen in such a way that 
	\begin{equation*}
		S_\varepsilon(g) = S(g)
	\end{equation*}
	for all $g \in G$.
\end{definition}
Thus, ansatz-posterior is just such a state from the ansatz family which coincides with the prior state for all relevant observables. 

For example, for non-selective measurement of an observable $X = X^{\dagger} \in \mathbb{C}^{d \times d}$ we have
\begin{equation*}
	G = \{g \in \mathbb{C}^{d \times d}: [X, g] = 0\}.
\end{equation*}
If one expands $X$ in the spectral expansion
\begin{equation*}
	X = \sum_{x} x \Pi_x, \qquad \Pi_x^{\dagger} = \Pi_x, \qquad \Pi_x \Pi_{x'} = \delta_{x,x'} \Pi_x, 
\end{equation*}
where $x$ are eigenvalues of $X $ and $\Pi_x$ are  projectors on eigenspaces correspondent to them. Then
\begin{equation*}
	G = \bigoplus\limits_x \mathbb{C}^{d_x \times d_x}, 
\end{equation*}
where $d_x$ is dimension of the  eigenspace correspondent to $x$. The states on $G$ can be parameterized by the matrices
\begin{equation*}
	\sigma_x \in \mathbb{C}^{d_{x} \times d_{x}},  \qquad \sigma_x = \sigma_x^\dagger \geqslant 0, \qquad \sum_x \operatorname{Tr}\sigma_x  = 1.
\end{equation*}
For $g = \oplus_x g_x$ we have
\begin{equation*}
	\varepsilon(g) = \sum_x \operatorname{Tr} g_x \sigma_x.
\end{equation*}
Then the ansatz in terms of parameters $ \sigma_x$ takes the form 
\begin{equation*}
	\rho_{ans}(\{\sigma_x\}) = \bigoplus\limits_x \sigma_x.
\end{equation*}
The ansatz-posterior for this  ansatz is the posterior state for a non-selective measurement.

For a selective measurement of $X$ with the result $x$ we have
\begin{equation*}
	G= \mathbb{C}^{d_{x} \times d_{x}} \oplus \mathbb{C} I_{d - d_{x}}.
\end{equation*}
The states on $G$ can be parameterized by the matrices
\begin{equation*}
	\sigma_x \in \mathbb{C}^{d_{x} \times d_{x}},  \qquad \sigma_x = \sigma_x^\dagger \geqslant 0,  \qquad \operatorname{Tr} \sigma_x \leqslant 1.
\end{equation*}
For $g = g_x \oplus \lambda I_{d - d_{x}}$ we have
\begin{equation*}
	\varepsilon(g) = \operatorname{Tr} \sigma_x g_x + \lambda (1 - \operatorname{Tr} \sigma_x)
\end{equation*}
and the ansatz in terms of parameter $ \sigma_x$ takes the form 
\begin{equation*}
	\rho_{ans}(\sigma_x) = \frac{\sigma_x}{\operatorname{Tr} \sigma_x} \oplus 0 I_{d - d_{x}}.
\end{equation*}

For the collision models for each collision the system and a new bath are prepared in the factorized state. Despite the fact that formally at each collision the bath, but if interaction between the  system and the bath during the evolution between collisions is assumed the same for each bath, then effectively it is prepared in the ansatz-posterior state for ansatz \eqref{eq:factorizedAnsatz}.

For the coordinate form of the ansatz parametrization \eqref{eq:vectorParam} we can give a definition equivalent to Def.~\ref{def:posteriors}, which is more useful for thermodynamic applications and master equation derivation for non-linear  ansatzes. 
\begin{definition}
	Let $\rho \in \mathbb{C}^{d \times d}$ be a density matrix and $\rho_{ans}(\vec{E})$ be an ansatz in $ \mathbb{C}^{d \times d}$ consistent with relevant observables $\vec{P} $. Then we call the density matrix
	\begin{equation}\label{eq:ansatzPosterior}
		\rho_{ans}(\vec{E})|_{\vec{E} =  \operatorname{Tr} \vec{P} \rho}
	\end{equation}
	an ansatz-posterior.
\end{definition}

\section{Stroboscopic limit for general ansatzes}
\label{sec:strob}

In this section we explain the setup of repeated measurements in the stroboscopic limit \cite{Luchnikov2017} for general ansatzes. 
Let us define repeated ansatz-posterior dynamics. It consists of two stages, which are repeated again and again:
\begin{enumerate}
	\item Instant transformation into the ansatz-posterior defined by Eq. \eqref{eq:ansatzPosterior}.
	\item Time evolution during the time $\Delta t$ with coupling $\lambda$. It is defined by a semigroup $e^{\lambda \mathcal{L} \Delta t}$ with some generator $ \mathcal{L} $ (in our examples we will assume that it is a GKSL generator).
\end{enumerate}
If  at some time-moment  the instant transformation into the ansatz-posterior had the parameter of ansatz $\vec{E} - \frac12 \Delta \vec{E}$, then after time $\Delta t$ and subsequent instant transformation into the ansatz-posterior we have the following  parameter of the ansatz
\begin{equation}\label{eq:strStep}
	\vec{E}(t) +  \frac12 \Delta \vec{E}(t) = 	 \operatorname{Tr} \vec{P} e^{\lambda \Delta t \mathcal{L}} \rho_{ans} \left(\vec{E}(t) - \frac12 \Delta \vec{E}(t) \right).
\end{equation}
In the following we will omit the time-dependence.

The stroboscopic limit for general ansatzes is 
\begin{equation}\label{eq:strobLimit}
	\lambda^2 \Delta t = \alpha = \text{fixed}, \qquad \Delta t \rightarrow 0.
\end{equation}

\begin{theorem}
	Let us assume that $\rho_{ans} (\vec{E}) $ is a continuously differentiable function of $\vec{E}$. Then in the stroboscopic limit \eqref{eq:strobLimit} the function $\vec{E}(t)$ satisfies the equation
	\begin{equation}\label{eq:secondOrderEq}
		\frac{d \vec{E}}{d t} = \lambda \langle  \vec{A} \rangle_{\vec{E}}    +  \frac12 \alpha \left(\langle  \vec{B} \rangle_{\vec{E}}   -     \left( \langle  \vec{A} \rangle_{\vec{E}}   , \frac{\partial}{\partial \vec{E}} \right) \langle  \vec{A} \rangle_{\vec{E}}  \right),
	\end{equation}
	where the terms $o(\Delta t)$ are omitted and
	\begin{equation}\label{eq:adjActDef}
		\mathcal{L}^*(\vec{P} ) \equiv \vec{A} , \qquad 	\mathcal{L}^*(\vec{A} ) \equiv \vec{B}, \qquad \langle \; \cdot \; \rangle_{\vec{E}} \equiv  \operatorname{Tr} (\; \cdot \; \rho_{ans}(\vec{E})),
	\end{equation}
	$\mathcal{L}^*$ is the adjoint projector (with respect to Frobenius scalar product) calculated element-wisely on the vector of matrices in its argument. The scalar product $  \left( \langle  \vec{A} \rangle_{\vec{E}}   , \frac{\partial}{\partial \vec{E}} \right)$ is understood as $ \left( \langle  \vec{A} \rangle_{\vec{E}}   , \frac{\partial}{\partial \vec{E}} \right) \equiv \sum_{j=1}^M\langle  A_j \rangle_{\vec{E}}  \frac{\partial}{\partial E_j}$.
\end{theorem}

	\begin{proof}
	Let us expand the right-hand side of Eq.~\eqref{eq:strStep} as
	\begin{align*}
		\operatorname{Tr} \vec{P} e^{\lambda \Delta t \mathcal{L}} \rho_{ans} \left(\vec{E} - \frac12 \Delta \vec{E} \right) =  \operatorname{Tr} \vec{P} \left(1 + \Delta t \lambda \mathcal{L} +  \frac12 \lambda^2  \Delta t^2 \mathcal{L}^2\right) \rho_{ans} \left(\vec{E} - \frac12 \Delta \vec{E} \right) + o(\Delta t^2)\\
		= \vec{E} - \frac12 \Delta \vec{E} +   \Delta t \lambda  \operatorname{Tr} \vec{P} \mathcal{L} \rho_{ans}(\vec{E}) -   \frac12 \Delta t \lambda   \operatorname{Tr} \vec{P} \mathcal{L}  \rho_{ans}\left( \Delta \vec{E}, \frac{\partial \rho_{ans} (\vec{E})}{\partial \vec{E}} \right) \\
		+ \frac12 \lambda^2  \Delta t^2  \operatorname{Tr} \vec{P} \mathcal{L}^2 \rho_{ans} (\vec{E})  + o(\Delta t^2 + \Delta E^2).
	\end{align*}
	By equating with the left-hand side of 	Eq.~\eqref{eq:strStep} and taking into account Eq.~\eqref{eq:strobLimit} we have
	\begin{equation*}
		\frac{\Delta \vec{E}}{\Delta t}  =  \lambda  \operatorname{Tr} \vec{P} \mathcal{L} \rho_{ans}(\vec{E}) -  \frac12 \frac{\alpha}{\lambda}   \operatorname{Tr} \vec{P} \mathcal{L} \left( \frac{\Delta \vec{E}}{\Delta t }, \frac{\partial \rho_{ans} (\vec{E})}{\partial \vec{E}} \right) +  \frac12 \alpha \operatorname{Tr} \vec{P} \mathcal{L}^2 \rho_{ans} (\vec{E})  + o(\Delta t).
	\end{equation*}
	Let us collect terms with $\frac{\Delta \vec{E}}{\Delta t}$ as
	
	\begin{equation*}
		\left(	I +  \frac12 \frac{\alpha}{\lambda}   \operatorname{Tr} \vec{P} \mathcal{L} \left( \; \cdot \;, \frac{\partial \rho_{ans} (\vec{E})}{\partial \vec{E}} \right)\right)\frac{\Delta \vec{E}}{\Delta t}  =  \lambda  \operatorname{Tr} \vec{P} \mathcal{L} \rho_{ans}(\vec{E})   +  \frac12 \alpha \operatorname{Tr} \vec{P} \mathcal{L}^2 \rho_{ans} (\vec{E})   + o(\Delta t).
	\end{equation*}
	Here $I$ is an $M \times M$ identity matrix acting on vectors of parameters and the next term near at the left-hand contains an $M \times M$ matrix which is defined for any $ \vec{X} \in \mathbb{C}^M $ as
	\begin{equation*}
		\operatorname{Tr} \vec{P} \mathcal{L} \left( \; \cdot \;, \frac{\partial \rho_{ans} (\vec{E})}{\partial \vec{E}} \right) \vec{X} =  \operatorname{Tr} \vec{P} \mathcal{L} \left( \vec{X}, \frac{\partial \rho_{ans} (\vec{E})}{\partial \vec{E}} \right).
	\end{equation*}
	As for $\lambda \rightarrow \infty$
	\begin{equation*}
		\left(	I +  \frac12 \frac{\alpha}{\lambda}   \operatorname{Tr} \vec{P} \mathcal{L} \left( \; \cdot \;, \frac{\partial \rho_{ans} (\vec{E})}{\partial \vec{E}} \right)\right)^{-1} =  	I -  \frac12 \frac{\alpha}{\lambda}   \operatorname{Tr} \vec{P} \mathcal{L} \left( \; \cdot \;, \frac{\partial \rho_{ans} (\vec{E})}{\partial \vec{E}} \right) + o(\lambda^{-1}),
	\end{equation*}
	then we have
	\begin{equation*}
		\frac{\Delta \vec{E}}{\Delta t} = \lambda  \operatorname{Tr} \vec{P} \mathcal{L} \rho_{ans}(\vec{E})   +  \frac12 \alpha \operatorname{Tr} \vec{P} \mathcal{L}^2 \rho_{ans} (\vec{E})  -  \frac12 \alpha   \operatorname{Tr} \vec{P} \mathcal{L} \left(  \operatorname{Tr} (\vec{P} \mathcal{L} \rho_{ans}(\vec{E})) , \frac{\partial \rho_{ans} (\vec{E})}{\partial \vec{E}} \right) + o(\Delta t).
	\end{equation*}
	Taking the limit $ \Delta t \rightarrow 0 $ and the definition of adjoint $ \operatorname{Tr} P \mathcal{L}  Y =  \operatorname{Tr} \mathcal{L}^*(P)  Y$ for any $Y \in \mathbb{C}^{d \times d}$ and any $P = P^{\dagger} \in \mathbb{C}^{d \times d}$ we have
	\begin{equation*}
		\frac{d \vec{E}}{d t} = \lambda  \operatorname{Tr}  \mathcal{L}^*(\vec{P}) \rho_{ans}(\vec{E})   +  \frac12 \alpha \operatorname{Tr}( \mathcal{L}^*)^2 ( \vec{P})\rho_{ans} (\vec{E})  -  \frac12 \alpha   \operatorname{Tr}  \mathcal{L}^*(\vec{P})  \left(  \operatorname{Tr} ( \mathcal{L}^*(\vec{P}) \rho_{ans}(\vec{E})) , \frac{\partial \rho_{ans} (\vec{E})}{\partial \vec{E}} \right).
	\end{equation*}
	As $ \operatorname{Tr} ( \mathcal{L}^*(\vec{P}) \rho_{ans}(\vec{E}))$ is a vector of complex numbers rather than the vector of matrices and its components are $\vec{E}$-independent, then
	\begin{equation*}
		\operatorname{Tr}  \mathcal{L}^*(\vec{P})  \left(  \operatorname{Tr} ( \mathcal{L}^*(\vec{P}) \rho_{ans}(\vec{E})) , \frac{\partial \rho_{ans} (\vec{E})}{\partial \vec{E}} \right) =     \operatorname{Tr} \left( \mathcal{L}^*(\vec{P}) \rho_{ans}(\vec{E})) , \left(  \frac{\partial }{\partial \vec{E}} \right)   \operatorname{Tr} ( \mathcal{L}^*(\vec{P}) \rho_{ans}(\vec{E}) \right).
	\end{equation*}
	
	Then taking into account \eqref{eq:adjActDef} we obtain \eqref{eq:secondOrderEq}.
\end{proof}

Let us remark that in \cite{Teretenkov2024a, Teretenkov2024b} a wide classes of many-body GKSL-generators were found which leave the invariant  spaces of the Hamiltionian dynamics still invariant in the presence of dissipation. In particular it is the case for integrals of motion of Hamiltionian dynamics for such classes, which is important for thermodynamical applications.

So let us consider the special case of invariant spaces, i.e.
\begin{equation}\label{eq:invariantCond}
	\mathcal{L}^*(\vec{P} ) = L \vec{P}, 
\end{equation}
where $L \in \mathbb{C}^{M \times M}$ acting on $\vec{P}$ as  on $M$-denominational vector, e.g. 
\begin{equation*}
	\mathcal{L}^*(P_m) = \sum_{k=1}^M L_{mk} P_k.
\end{equation*}
In the coordinate-free form it means that $\mathcal{L}^*$ and, hence, the Heisenberg dynamics leaves the operator system $G$ invariant
\begin{equation*}
	\mathcal{L}^*(G) \subseteq G, \qquad e^{\mathcal{L}^* t}(G) \subseteq G.
\end{equation*}
From \eqref{eq:invariantCond} we have
\begin{equation*}
	\mathcal{L}^*(\vec{A} ) = L^2 \vec{P}.
\end{equation*}

Then Eq.~\eqref{eq:secondOrderEq} takes the form
\begin{equation}\label{eq:secondOrderEqInvariant}
	\frac{d \vec{E}}{d t} = \lambda   L \langle  \vec{P} \rangle_{\vec{E}}    +  \frac12 \alpha \left( L^2 \langle \vec{P} \rangle_{\vec{E}}   -     \left(  L\langle  \vec{P} \rangle_{\vec{E}}   , \frac{\partial}{\partial \vec{E}} \right)  L \langle  \vec{P} \rangle_{\vec{E}}  \right).
\end{equation}
But from Eq.~\eqref{eq:vectorParam} we have
\begin{equation*}
	\langle  \vec{P} \rangle_{\vec{E}} = \vec{E},
\end{equation*}
and then taking into account
\begin{equation*}
	\left(  L \vec{E}   , \frac{\partial}{\partial \vec{E}} \right)  L\vec{E} = L^2  \vec{E} 
\end{equation*}
Eq.~\eqref{eq:secondOrderEqInvariant} reduces just to
\begin{equation}\label{eq:zeroSecTerm}
	\frac{d \vec{E}}{d t} = \lambda   L \vec{E}.
\end{equation}

Thus, \eqref{eq:secondOrderEq} can be considered as perturbative correction to the dynamics with invariant subspace.

Another special case is the linear ansatz one. If $ \rho_{ans}(\vec{E})$ is linear in $\vec{E}$, then $\langle  \vec{A} \rangle_{\vec{E}}$ is linear in $\vec{E}$ as well, then 
\begin{equation*}
	\left( \langle  \vec{A} \rangle_{\vec{E}}   , \frac{\partial}{\partial \vec{E}} \right) \langle  \vec{A} \rangle_{\vec{E}} = \langle  \vec{A} \rangle_{\langle  \vec{A} \rangle_{\vec{E}}}.
\end{equation*}
and Eq.~\eqref{eq:secondOrderEq} takes the form
\begin{equation}\label{eq:linearSecOrdEq}
	\frac{d \vec{E}}{d t} = \lambda \langle  \vec{A} \rangle_{\vec{E}}    +  \frac12 \alpha \left(\langle  \vec{B} \rangle_{\vec{E}}   -     \langle  \vec{A} \rangle_{\langle  \vec{A} \rangle_{\vec{E}}} \right),
\end{equation}
and as $\lambda \langle  \vec{A} \rangle_{\vec{E}}$ and $\lambda \langle  \vec{B} \rangle_{\vec{E}}$ are linear functions in  $\vec{E}$, it is just a linear ordinary differential equation for $\vec{E}(t)$.

For parmeterization \eqref{eq:projAns} of linear ansatzes in terms of projectors  Eq.~\eqref{eq:linearSecOrdEq} takes the form
\begin{equation*}
	\frac{d }{d t} \mathcal{P} \rho(t) = \lambda \mathcal{P} \mathcal{L} \mathcal{P} \rho(t)    +  \frac12 \alpha \left(\mathcal{P}  \mathcal{L}^2\mathcal{P} \rho(t)   -   \mathcal{P} \mathcal{L} \mathcal{P} \mathcal{L}\mathcal{P} \rho(t) \right).
\end{equation*}
In this form it is fully similar to projection-based form of superoperator master equation in the stroboscopic limit \cite{Teretenkov2024, Teretenkov2024c}.

	\section{Generalized temperature measurement}
\label{sec:tempMeas}

The main aim of this work is to apply a general framework to quantum thermodynamics and especially to thermometry. So we start from some general setup which covers many quantum thermodynamical situations. 

Namely, we consider the generalized Gibbs ansatzes of the form
\begin{equation}\label{eq:GibbsFamily}
	\rho_{Gibbs}(\vec{\beta}) = \frac{e^{- (\vec{\beta}, \vec{P})}}{Z(\vec{\beta})}, \qquad Z(\vec{\beta}) \equiv \operatorname{Tr} e^{- (\vec{\beta}, \vec{P})}.
\end{equation}
Here $\vec{P}$ is the vector of relevant observables, which need not be commutative in general \cite{Hinds2018}. This ansatz is parameterized by $\vec{\beta}$, which are natural to be called generalized inverse temperatures. But this parametrization is not consistent with relevant observables, i.e. $\operatorname{Tr} \vec{P} \rho_{Gibbs}(\vec{\beta}) \neq \vec{\beta}$. Thus, to fit it in our approach we need to solve the equation
\begin{equation}\label{eq:consistTransform}
	\operatorname{Tr} \vec{P} \rho_{Gibbs}(\vec{\beta}) = \vec{E}
\end{equation}
with respect to a $\vec{\beta}$ for given $\vec{E}$. Its solution $\vec{\beta}(\vec{E})$ defines the ansatz
\begin{equation}\label{eq:GibbsFamilyReparam}
	\rho_{ans}(\vec{E}) = \rho_{Gibbs}(\vec{\beta}(\vec{E})),
\end{equation}
which satisfies the consistency conditions and, hence,  Eq.~\eqref{eq:secondOrderEq} is applicable to it.

Let us remark that taking into account \eqref{eq:GibbsFamily} one can rewrite Eq.~\eqref{eq:consistTransform} as
\begin{equation*}
	\vec{E}=-\frac{\partial}{\partial \vec{\beta}} \ln Z(\vec{\beta}),
\end{equation*}
which makes the problem of reparametrization of \eqref{eq:GibbsFamily} essentially the problem of equilibrium thermodynamics.

Now we focus on canonical Gibbs state for which Eq.~\eqref{eq:GibbsFamily} reduces to
\begin{equation}\label{eq:GibbsAns}
	\rho_{Gibbs}(\beta) = \frac{e^{- \beta H}}{Z(\beta)}.
\end{equation}
In terms of Eq.~\eqref{eq:GibbsFamily} here the vector $\vec{\beta}$ has only one component $\beta$ and  the vector $\vec{P}$ has only one component $H$. Then Eq.~\eqref{eq:secondOrderEq} takes the form
\begin{equation}\label{eq:scalarSecOrd}
	\frac{d E}{d t} = \lambda \langle  A \rangle_{E}    +  \frac12 \alpha \left(\langle  B \rangle_{E}   -   \frac12   \frac{\partial}{\partial E} \langle  A \rangle_{E}^2  \right).
\end{equation}

As \eqref{eq:GibbsAns} is parameterized in terms of temperature and, moreover, we are interested in results of temperature measurement, then it can be useful to  rewrite this equation directly in terms of temperature:
\begin{equation}\label{eq:invTempSecOrd}
	\frac{d \beta}{d t}  = \lambda \langle  A \rangle_{E(\beta)}  \frac{\partial \beta}{\partial E}  +  \frac12 \alpha \left(\langle  B \rangle_{E(\beta)}   \frac{\partial \beta}{\partial E} -   \frac12  \left(\frac{\partial \beta}{\partial E} \right)^2  \frac{\partial}{\partial \beta} \langle  A \rangle_{E(\beta)}^2  \right).
\end{equation}
One can also express the derivative 
\begin{equation*}
	\frac{\partial \beta}{\partial E} = - \frac{\beta^2}{C(\beta)}
\end{equation*}
in terms of the heat capacity
\begin{equation*}
	C(\beta) = \frac{\partial E}{\partial( \beta^{-1})},
\end{equation*}
which can be regarded as a ''standard'' function for equilibrium thermodynamics, and it can be useful to use the known results from the literature if one applies this equation to quantum many-body systems.

\section{Examples: Qubit and generic multi-level system as a probe thermometer}
\label{eq:example}

As a first example let us describe the open two-level system in the external field as a thermometer for its environment. We will assume that its dynamics is described by GKSL equation for resonant fluorescence \cite[Sec.~3.4.5]{Breuer2002} with the generator
\begin{align*}
	\mathcal{L}(\rho) =&  -i \left[  ( \omega_0 + \Delta \omega) \sigma_+\sigma_- -  \Omega \left( \sigma_++\sigma_-\right), \rho \right] \\
	&+  \gamma \left( \sigma_- \rho\sigma_+ - \frac12 \{ \sigma_+\sigma_-, \rho \}+e^{-\beta_0\omega_0} \left( \sigma_+ \rho \sigma_- - \{ \sigma_+\sigma_-, \rho \} \right) \right),
\end{align*}
where $\omega_0$ is the transition frequency of the two-level system, $\Delta\omega$ is the Lamb shift, $ \Omega$ is the Rabi frequency, $\beta_0$ is the inverse temperature, $\gamma > 0$ is the parameter describing the dissipation rate. Let us remark that we have isolated in the equation only the dependence on $\beta_0$, which is defined by detailed balance  and is universal for the arbitrary reservoir in the equilibrium state, but $\gamma$  still can be a function of $\beta_0$ which depends on a specific model of the reservoir. In particular, for the bosonic reservoir one has
\begin{equation*}
	\gamma = \frac{e^{\beta_0 \omega_0}}{e^{\beta_0 \omega_0} - 1} \gamma_0,
\end{equation*}
where $ \gamma_0$ is a temperature-independent constant. Remark that in this generator we assume that the external field is undepleted, but the depletion can be taken into account by the method developed in \cite{Karasev2024}.

We assume that the repeated  measurement of the two-level system is performed, so it can be considered as a probe thermometer. We consider the Gibbs ansatz \eqref{eq:GibbsAns} with
\begin{equation*}
	H = \omega_0 \sigma_+ \sigma_-.
\end{equation*}

From Eq.~\eqref{eq:consistTransform} $E$ can be written as:
\begin{equation*}
	E = \frac{\omega_0  e^{-\beta \omega_0 }}{e^{-\beta \omega_0 }+1},
\end{equation*}
which gives
\begin{equation}\label{eq:betaE}
	\beta(E) = -\frac{1}{\omega_0}\ln{\frac{E}{\omega_0  -E}},
\end{equation}
then Eq.~\eqref{eq:GibbsFamilyReparam} takes the form
\begin{equation}\label{eq:linearAnsatz}
	\rho_{ans}(E) = \frac{E}{\omega_0} \sigma_+ \sigma_- + \left( 1- \frac{E}{\omega_0} \right) \sigma_- \sigma_+.
\end{equation}
In this case $\langle  A \rangle_{E}$ and $\langle  B\rangle_{E}$ from Eq.~\eqref{eq:scalarSecOrd} take the form:
\begin{align*}
	\langle  A \rangle_{E} &= - \gamma (1 + e^{- \beta_0 \omega_0}) E + \gamma  \omega_0 e^{- \beta_0 \omega_0}, 
	\\
	\langle  B \rangle_{E} &= \gamma ^2 e^{-2 \beta _0 \omega _0} \left(e^{\beta _0 \omega _0}+1\right) \left(E e^{\beta _0 \omega _0}+E-\omega _0\right)+2 \Omega ^2 \left(\omega _0-2
	E\right).
\end{align*}
Thus, Eq.~\eqref{eq:GibbsFamilyReparam} takes the form
\begin{equation}\label{eq:qubitSecOrd}
	\frac{d E}{d t} = - \gamma (1 + e^{- \beta_0 \omega_0}) E + \gamma  \omega_0 e^{- \beta_0 \omega_0} + 2 \Omega^2\Delta t (\omega _0-2
	E),
\end{equation}
where we have omitted $\lambda$ absorbing it into the constants of the generator, and $\alpha$ becomes just the time between measurements $\Delta t$.   Remark that the last term becomes zero in the case $\Omega = 0$. This is due to the fact that the linear span of the identity matrix $I$   Hamiltonian $H$ is invariant space for the adjoint GKSL generator in this case and, hence, fits Eq.~\eqref{eq:zeroSecTerm} after linear change of the relevant observable $H$.

The solution of Eq.~\eqref{eq:qubitSecOrd} has the form
\begin{equation}\label{eq:Et}
	E(t) = e^{- \left(\gamma  \left(1 + e^{-\beta _0 \omega _0}\right)+4 \Omega^2 \Delta t  \right) t} (E(0) - E_{st}) +  E_{st},
\end{equation}
where
\begin{equation}\label{eq:stationaryE}
	E_{st} = \omega_0\frac{ \gamma   e^{- \beta_0 \omega_0} + 2 \Omega^2\Delta t}{\gamma  \left(1 + e^{-\beta _0 \omega _0}\right)+4 \Omega^2 \Delta t }.
\end{equation}

Substituting Eq.~\eqref{eq:stationaryE} into Eq.~\eqref{eq:betaE} we have
\begin{equation*}
	\beta_{st} = \beta_0 - \frac{1}{\omega_0} \ln  \frac{  1 + \frac{2 \Omega^2\Delta t}{\gamma} e^{\beta_0 \omega_0}}{1 + \frac{2 \Omega^2\Delta t}{\gamma}}.
\end{equation*}
And
\begin{equation*}
	\tau =(\gamma  \left(1 + e^{-\beta _0 \omega _0}\right)+4 \Omega^2 \Delta t )^{-1}
\end{equation*}
in the exponential function in  Eq.~\eqref{eq:Et} can be considered as time of measurement of the temperature of the external environment. Thus, in presence of the external field the temperature measurement becomes less precise but faster. So the trade off between this factors should be taken into account in experiments.

We emphasize that the linearity of the Gibbs ansatz \eqref{eq:linearAnsatz}, which lead to linearity of Eq.~\eqref{eq:qubitSecOrd} is specific for the two-level system \cite[Sec. 4]{Meretukov2024}.

Let us also consider  an example of a generic multi-level weak coupling limit type generator \cite{Accardi2002}. Unlike the qubit system we do not introduce the external classical field, but even in this case the equation is highly nonlinear. It has the following form
\begin{equation*}
	\mathcal{L}(\rho) = - i \left[ \sum_j( \omega_j + \Delta \omega_j) |j\rangle \langle j|, \rho \right] + \sum_{ij} \gamma_{ij} \left(|i \rangle \langle j| \rho |j\rangle \langle i| - \frac12 \{ |j \rangle \langle j|,\rho \} \right),
\end{equation*}
where $|j\rangle, j= 1, \ldots, d$ is the eigenbasis the system Hamiltonain, $\gamma_{ij} \geqslant 0$. We consider the Gibbs ansatz Eq.~\eqref{eq:GibbsAns} with
\begin{equation*}
	H = \sum_j  \omega_j  |j\rangle \langle j|.
\end{equation*}

Then  $\langle  A \rangle_{E}$ and $\langle  B\rangle_{E}$ take the form:
\begin{align*}
	\langle  A \rangle_{E(\beta)} &=  \frac{1}{Z( \beta)} \sum_{ij}  \omega_i (\gamma_{ij}  e^{- \beta \omega_j} -  \gamma_{ji}  e^{- \beta \omega_i}),  \\
	\langle  B\rangle_{E(\beta)} 
	&= \frac{1}{Z( \beta)}\sum_{ikj} ( e^{- \beta \omega_j}\gamma_{ij} -  e^{- \beta \omega_i}\gamma_{ji}) \gamma_{ki}  ( \omega_k  -  \omega_i) .
\end{align*}
Under detailed balance conditions \cite{Alicki1976, Fagnola2010} we have
\begin{equation*}
	\gamma_{ji} =  \gamma_{ij} e^{- \beta_0(\omega_i - \omega_j)},
\end{equation*}
so then we obtain
\begin{align*}
	\langle  A \rangle_{E(\beta)} &=  \frac{1}{Z( \beta)} \sum_{ij}  \omega_i  e^{- \beta \omega_j} \gamma_{ij}(1  -    e^{- (\beta - \beta_0)(\omega_i - \omega_j)})  ,
	\\
	\langle  B\rangle_{E(\beta)} 
	&= \frac{1}{Z( \beta)}\sum_{ikj}  e^{- \beta \omega_j} \gamma_{ij}(1  -    e^{- (\beta - \beta_0)(\omega_i - \omega_j)} )\gamma_{ki}  ( \omega_k  -  \omega_i) .
\end{align*}
In this form they can be they can be substituted into Eq~\eqref{eq:invTempSecOrd} and it is easily seen that $\beta_{st} =  \beta_0$ is its  stationary solutions, so it can model  asymptotically exact indirect measurement of the temperature. 

\section{Conclusions}

In this work we have developed a very general framework describing stroboscopic limit of repeated measurements for arbitrary ansatzes. Similarly to weak coupling limit ansatz-based master equations the most instant area of application of such an approach is quantum thermodynamics.  In the same sense that for the weak coupling case the Gibbs ansatz-based approach is called ''thermodynamical approach to open quantum systems'' \cite{Semin2020, Meretukov2024}, we can call its stroboscopic counterpart developed here ''thermodynamical  approach to repeated measurements''. So we have applied our framework to termometry, because the temperature is one of the most paradigmatic ''thermodynamical'' observable, which have no obvious place in the usual framework of von Neumann measurements. Nevertheless, we tried not only to achieve this specific goal, but to reveal the deep similarities between different approaches. These similarities become much clearer in the general framework and hidden if one focus on its specific realizations.

As a possible direction for further studies, we should mention the comparison of our ansatz-based approach with those based on parameter estimation \cite{Albarelli2023Invasiveness, DePasquale2017Sequential}. It would also be interesting to combine our approach with collision model derivation \cite{Campbell2021, Ciccarello2022} with total unitary dynamics of the system and the colliding baths, rather than just postulating GKSL dynamics of the system. This should still fit our general framework, but it also allows considering the regimes, where one measures the system temperature faster than its dynamics becomes Markovian. Another natural direction for further studies is the calculation of  the memory tensors \cite{Teretenkov2023} defining the multi-time correlation functions. The generalization of the present results to the infinite-dimensional dynamics with unbounded generator similar to the one recently done  \cite{Lopatin2025} for the usual weak cooling limit also, will be interesting.


\begin{thebibliography}{99}
	
	\bibitem{Albarelli2023Invasiveness}
	% \refitem{article}
	F.~Albarelli, M. G. A.~Paris, B.~Vacchini, and A.~Smirne, \textquotedblleft Invasiveness of nonequilibrium pure-dephasing quantum thermometry,\textquotedblright\;Phys. Rev. A \textbf{108} (6) (2023).
	
	\bibitem{DePasquale2017Sequential}
	% \refitem{article}
	A.~De Pasquale, K.~Yuasa, and V.~Giovannetti, \textquotedblleft Estimating temperature via sequential measurements,\textquotedblright\;Phys. Rev. A \textbf{96} (1) (2017).
	
	\bibitem{Sparaciari2017}
	% \refitem{article}
	C.~Sparaciari, J.~Oppenheim, and T.~Fritz, \textquotedblleft Resource theory for work and heat,\textquotedblright\;Phys. Rev. A \textbf{96} (5), 052112 (2017).
	
	\bibitem{Meretukov2024}
	% \refitem{article}
	Kh. Sh.~Meretukov and A. E.~Teretenkov, \textquotedblleft On Time-Dependent Projectors and a Generalization of the Thermodynamical Approach in the Theory of Open Quantum Systems,\textquotedblright\;Proc. Steklov Inst. Math. \textbf{324}, 135--152 (2024).
	
	\bibitem{Luchnikov2017}
	% \refitem{article}
	I. A.~Luchnikov and S. N.~Filippov, \textquotedblleft Quantum evolution in the stroboscopic limit of repeated measurements,\textquotedblright\;Phys. Rev. A \textbf{95}, 022113 (2017).
	
	\bibitem{Rau1996}
	% \refitem{article}
	J.~Rau and B.~Müller, \textquotedblleft From reversible quantum microdynamics to irreversible quantum transport,\textquotedblright\;Phys. Rep. \textbf{272} (1), 1--59 (1996).
	
	\bibitem{Semin2015}
	% \refitem{article}
	V.~Semin and F.~Petruccione, \textquotedblleft Projection Operators in the Theory of Open Quantum Systems,\textquotedblright\;Proc. SAIP, 539--544 (2015).
	
	\bibitem{Semin2020}
	% \refitem{article}
	V.~Semin and F.~Petruccione, \textquotedblleft Dynamical and Thermodynamical Approaches to Open Quantum Systems,\textquotedblright\;Sci. Rep. \textbf{10} (1), 2607 (2020).
	
	\bibitem{Gorini1976}
	% \refitem{article} 
	V.~Gorini, A.~Kossakowski, E.C.G.~Sudarshan, \textquotedblleft Completely positive dynamical semigroups of N-level systems\textquotedblright\;J. of Math. Phys., \textbf{17} (5), 821--825 (1976).
	
	\bibitem{Lindblad1976}
	% \refitem{article} 
	G.~Lindblad, \textquotedblleft On the generators of quantum dynamical semigroups,\textquotedblright\;Comm. in Math. Phys., \textbf{48} (2), 119--130 (1976).
	
	\bibitem{Choi1977}
	% \refitem{article}
	M.-D.~Choi and E.~Effros, \textquotedblleft Injectivity and operator spaces,\textquotedblright\;J. Funct. Anal. \textbf{24} (2), 156--209 (1977).
	
	\bibitem{Paulsen2003}
	% \refitem{book}
	V.~Paulsen, \textquotedblleft Completely bounded maps and operator algebras,\textquotedblright\;Cambridge Univ. Press (2003).
	
	\bibitem{Yashin2022}
	% \refitem{article}
	V. I.~Yashin, \textquotedblleft The extension of unital completely positive semigroups on operator systems to semigroups on \(C^*\)-algebras,\textquotedblright\;Lobachevskii J. Math. \textbf{43} (7), 1778--1790 (2022).
	
	\bibitem{Duan2012}
	% \refitem{article}
	R.~Duan, S.~Severini, and A.~Winter, \textquotedblleft Zero-error communication via quantum channels, noncommutative graphs, and a quantum Lovász number,\textquotedblright\;IEEE Trans. Inf. Theory \textbf{59} (2), 1164--1174 (2012).
	
	\bibitem{Amosov2019}
	% \refitem{article}
	G. G.~Amosov and A. S.~Mokeev, \textquotedblleft On linear structure of non-commutative operator graphs,\textquotedblright\;Lobachevskii J. Math. \textbf{40}, 1440--1443 (2019).
	
	\bibitem{Amosov2020}
	% \refitem{article}
	G. G.~Amosov, A. S.~Mokeev, and A. N.~Pechen, \textquotedblleft Non-commutative graphs and quantum error correction for a two-mode quantum oscillator,\textquotedblright\;Quantum Inf. Process. \textbf{19}, 1--12 (2020).
	
	\bibitem{Amosov2021}
	% \refitem{article}
	G. G.~Amosov, A. S.~Mokeev, and A. N.~Pechen, \textquotedblleft Noncommutative graphs based on finite-infinite system couplings: Quantum error correction for a qubit coupled to a coherent field,\textquotedblright\;Phys. Rev. A \textbf{103} (4), 042407 (2021).
	
	\bibitem{Weaver2021}
	% \refitem{article}
	N.~Weaver, \textquotedblleft Quantum graphs as quantum relations,\textquotedblright\;J. Geom. Anal. \textbf{31} (9), 9090--9112 (2021).
	
	\bibitem{Brannan2022}
	% \refitem{article}
	M.~Brannan, P.~Ganesan, and S. J.~Harris, \textquotedblleft The quantum-to-classical graph homomorphism game,\textquotedblright\;J. Math. Phys. \textbf{63} (11) (2022).
	
	\bibitem{Priyanga2023}
	% \refitem{article}
	P.~Ganesan, \textquotedblleft Spectral bounds for the quantum chromatic number of quantum graphs,\textquotedblright\;Linear Algebra Appl. \textbf{674}, 351--376 (2023).
	
	\bibitem{Klerk2011}
	% \refitem{article}
	E.~De Klerk, C.~Dobre, and D. V.~Pasechnik, \textquotedblleft Numerical block diagonalization of matrix*-algebras with application to semidefinite programming,\textquotedblright\;Math. Program. \textbf{129}, 91--111 (2011).
	
	\bibitem{Breuer2002}	
	% \refitem{book}
	H.-P.~Breuer, F.~Petruccione, \emph{The theory of open quantum systems} (Oxford University Press, Oxford, 2002).
	
	\bibitem{Pechukas1994}
	% \refitem{article}
	P.~Pechukas, \textquotedblleft Reduced dynamics need not be completely positive,\textquotedblright\;Phys. Rev. Lett. \textbf{73} (8), 1060--1062 (1994).
	
	\bibitem{Alicki1995}
	% \refitem{article}
	R.~Alicki, \textquotedblleft Comment on ``Reduced dynamics need not be completely positive'',\textquotedblright\;Phys. Rev. Lett. \textbf{75} (16), 3020 (1995).
	
	\bibitem{Jordan2004}
	% \refitem{article}
	T. F.~Jordan, A.~Shaji, and E. C. G.~Sudarshan, \textquotedblleft Dynamics of initially entangled open quantum systems,\textquotedblright\;Phys. Rev. A \textbf{70} (5), 052110 (2004).
	
	\bibitem{Rodriguez2010}
	% \refitem{article}
	C. A.~Rodríguez-Rosario, K.~Modi, and A.~Aspuru-Guzik, \textquotedblleft Linear assignment maps for correlated system-environment states,\textquotedblright\;Phys. Rev. A \textbf{81} (1), 012313 (2010).
	
	\bibitem{Trushechkin2021}
	% \refitem{article}
	A. S.~Trushechkin, \textquotedblleft Derivation of the Redfield Quantum Master Equation and Corrections to It by the Bogoliubov Method,\textquotedblright\;Proc. Steklov Inst. Math. \textbf{313}, 246--257 (2021).
	
	\bibitem{Teretenkov2019}
	% \refitem{article}
	A. E.~Teretenkov, \textquotedblleft Irreversible quantum evolution with quadratic generator: Review,\textquotedblright\;Infin. Dimens. Anal. Quantum Probab. Relat. Top. \textbf{22} (4), 19300019 (2019).
	
	\bibitem{Meretukov2025}
	% \refitem{article}
	Kh. Sh.~Meretukov and A. E.~Teretenkov, \textquotedblleft Gaussian approximation and its corrections for driven dissipative Kerr model,\textquotedblright\;arXiv:2410.09547 (2024).
	
	\bibitem{Zubarev1997}
	% \refitem{book}
	D.~Zubarev, V. G.~Morozov, and G.~Röpke, \emph{Statistical Mechanics of Nonequilibrium Processes} (Akademie Verlag, Berlin, 1997).
	
	
	\bibitem{Campbell2021}
	% \refitem{article}
	S.~Campbell and B.~Vacchini, \textquotedblleft Collision models in open system dynamics: A versatile tool for deeper insights?,\textquotedblright\;Europhys. Lett. \textbf{133} (6), 60001 (2021).
	
	\bibitem{Ciccarello2022}
	% \refitem{article}
	F.~Ciccarello, S.~Lorenzo, V.~Giovannetti, and G. M.~Palma, \textquotedblleft Quantum collision models: Open system dynamics from repeated interactions,\textquotedblright\;Phys. Rep. \textbf{954}, 1--70 (2022).
	
	\bibitem{Teretenkov2024a}
	% \refitem{article}
	A.~Teretenkov and O.~Lychkovskiy, \textquotedblleft Exact dynamics of quantum dissipative XX models: Wannier--Stark localization in the fragmented operator space,\textquotedblright\;Phys. Rev. B \textbf{109} (14), 140302 (2024).
	
	\bibitem{Teretenkov2024b}
	% \refitem{article}
	A.~Teretenkov and O.~Lychkovskiy, \textquotedblleft Duality between open systems and closed bilayer systems: Thermofield double states as quantum many-body scars,\textquotedblright\;Phys. Rev. B \textbf{110}, 241105 (2024).
	
	
	\bibitem{Teretenkov2024}
	% \refitem{article}
	A. E.~Teretenkov, \textquotedblleft Superoperator master equations and effective dynamics,\textquotedblright\;Entropy \textbf{26} (1), 14 (2024).
	
	\bibitem{Teretenkov2024c}
	% \refitem{article}
	A. E.~Teretenkov, \textquotedblleft Superoperator master equations for depolarizing dynamics,\textquotedblright\;Lobachevskii J. Math. \textbf{45} (6) (2024).

	
	\bibitem{Hinds2018}
	% \refitem{incollection}
	M. E.~Hinds, Y.~Guryanova, P.~Faist, and D.~Jennings, \textquotedblleft Quantum Thermodynamics with Multiple Conserved Quantities,\textquotedblright\;in \emph{Thermodynamics in the Quantum Regime: Fundamental Aspects and New Directions} (Springer, Berlin, 2018). 
	
	\bibitem{Karasev2024}
	% \refitem{article}
	A. Y.~Karasev and A. E.~Teretenkov, \textquotedblleft Parametric approximation as open quantum systems problem,\textquotedblright\;arXiv:2410.03482 (2024).
	
	\bibitem{Accardi2002}
	% \refitem{incollection}
	L.~Accardi and S.~Kozyrev, \textquotedblleft Lectures on quantum interacting particle systems,\textquotedblright\;in \emph{Quantum Interacting Particle Systems (Trento, 2000)}, QP--PQ: Quantum Probab. White Noise Anal. \textbf{14}, 1195 (World Scientific, Singapore, 2002).
	
	\bibitem{Alicki1976}
	% \refitem{article}
	R.~Alicki, \textquotedblleft On the detailed balance condition for non-Hamiltonian systems,\textquotedblright\;Rep. Math. Phys. \textbf{10} (2), 249--258 (1976).
	
	\bibitem{Fagnola2010}
	% \refitem{article}
	F.~Fagnola and V.~Umanità, \textquotedblleft On two quantum versions of the detailed balance condition,\textquotedblright\;Banach Center Publ. \textbf{89} (1), 105--119 (2010).
	
	\bibitem{Teretenkov2023}
	% \refitem{article}
	A. E.~Teretenkov, \textquotedblleft Memory tensor for non-Markovian dynamics with random Hamiltonian,\textquotedblright\;Mathematics \textbf{11} (18), 3854 (2023).
	
	\bibitem{Lopatin2025}
	% \refitem{article}
	I. A.~Lopatin and A. N.~Pechen, \textquotedblleft Weak coupling limit for quantum systems with unbounded weakly commuting system operators,\textquotedblright\;J. Math. Phys. \textbf{66}, 042101 (2025).
	

\end{thebibliography}
\end{document}